\documentclass[a4paper, reqno]{amsart}
\usepackage{color}
\usepackage{amsmath}
\usepackage{amsthm}
\usepackage{amssymb}
\usepackage{amsmath}
\usepackage{mathcomp}
\usepackage{dsfont}
\usepackage[toc,page]{appendix}
\usepackage{graphicx} 
\usepackage{subfigure}
\usepackage{enumerate}
\usepackage{amsfonts}
\usepackage{amsthm}
\usepackage{amsmath}
\usepackage{times}
\usepackage{amscd}
\usepackage[latin2]{inputenc}
\usepackage{t1enc}
\usepackage{enumitem}
\usepackage[mathscr]{eucal}
\usepackage{indentfirst}
\usepackage{graphicx}
\usepackage{graphics}
\usepackage{pict2e}
\usepackage{epic}
\usepackage{esint}
\usepackage[margin=2.9cm]{geometry}
\usepackage{epstopdf} 
\usepackage{verbatim}
\usepackage{cancel}
\usepackage{amssymb}
\usepackage{xcolor}
\usepackage{dsfont}
\usepackage{hyperref}
\usepackage{physics}

\usepackage[margin=2.9cm]{geometry}

\newtheorem{definition}{Definition}[section]
\newtheorem{lemma}[definition]{Lemma}

\newtheorem{theorem}[definition]{Theorem}

\numberwithin{equation}{section}

\def\bR{\mathbb{R}}

\def\bP{\mathbb{P}}
\def\bE{\mathbb{E}}

\def\cN{\mathcal{N}}
\def\cU{\mathcal{U}}
\def\cW{\mathcal{W}}
\def\cF{\mathcal{F}}
\def\cM{\mathcal{M}}
\def\cJ{\mathcal{J}}
\def\cG{\mathcal{G}}
\def\cL{\mathcal{L}}

\def\b1{\mathds{1}}

\def\Re{\mathrm{Re}}
\def\Im{\mathrm{Im}}

\def\wO{\widetilde{O}}

\def\ph{\varphi}
\def\wt{\widetilde}

\newcommand{\vertiii}[1]{{\left\vert\kern-0.25ex\left\vert\kern-0.25ex\left\vert #1 
    \right\vert\kern-0.25ex\right\vert\kern-0.25ex\right\vert}}

\begin{document}
\title[Large deviation estimates for weakly interacting bosons]{Large deviation estimates for weakly interacting bosons}

\author{Simone Rademacher and Robert Seiringer}
\address{IST Austria, Am Campus 1, 3400 Klosterneuburg, Austria}

\date{\today}

\begin{abstract}
We study the many-body dynamics of an initially factorized bosonic wave function in the mean-field regime. We prove large deviation estimates for the fluctuations around the condensate. We derive an upper bound extending a recent result to more general interactions. Furthermore, we derive a new lower bound which agrees with the upper bound in leading order. 
\end{abstract}

\maketitle

\section{Introduction and main results}

\subsection{Introduction} We consider the dynamics of $N$ bosons in the mean-field regime described through the bosonic wave function $\psi_{N,t} \in L_{\rm s}^2 ( \bR^{3N})$, the symmetric subspace of $L^2 ( \bR^{3N} )$. The bosons evolve according to the Schr\"odinger equation 
\begin{align}
\label{eq:Schroe}
i \partial_t \psi_{N,t} = H_N \psi_{N,t} \; 
\end{align}
 where $H_N$ denotes the Hamiltonian 
\begin{align}
H_N = \sum_{j=1}^N - \Delta_{x_j} + \frac{1}{N} \sum_{i<j}^N v( x_i - x_j) \,.
\end{align}
The coupling constant $1/N$ in front of the interaction term corresponds to weak and long-range interactions of mean-field type. 
In the following we assume  the two-particle interaction potential $v$ to satisfy
\begin{align}
\label{eq:ass-v}
v^2 \leq C \left( 1 -\Delta\right)
\end{align}
 for a positive constant $C>0$.  We consider factorized initial data $\psi_{N,0} = \ph^{\otimes N}$ exhibiting complete Bose--Einstein condensation (BEC), i.e.  their reduced one-particle density $\gamma_N$ satisfies 
 \begin{align}
 \gamma_N =  \vert \varphi \rangle \langle \varphi \vert \quad \text{for every $N$,}
 \end{align}
 for a one-particle orbital $\varphi \in H^4 (\bR^3)$.  Although the factorization is not preserved along the time evolution, the property of BEC is known to be preserved, i.e. the reduced one-particle density $\gamma_{N,t}$ associated to the solution $\psi_{N,t}$ of the Schr\"odinger equation \eqref{eq:Schroe} satisfies 
 \begin{align}
 \label{eq:BEC}
 \gamma_{N,t} \rightarrow \vert \varphi_t \rangle \langle \varphi_t \vert \quad \text{as} \quad N \rightarrow \infty 
 \end{align}
where the time evolution of the condensate wave function  $\varphi_t$ is governed by the Hartree equation 
\begin{align}
\label{eq:hartree}
i \partial_t \ph_t = h_{\rm H}(t) \; \ph_t, \quad \text{with} \quad h_{\rm H} (t) = - \Delta +  v*\vert \ph_t \vert^2
\end{align}
 with initial data $\varphi_{0} = \varphi$. (For more details see e.g. \cite{AGT,AFP,BGM,CLS,C,EY,FKP,FKS,GV,KP,RoS,Sp}.)

\subsection{Main results} From a probabilistic point of view,  BEC implies a law of large numbers for bounded one-particle observables. To be more precise,  for a bounded, self-adjoint one-particle operator $O$ on $L^2 ( \mathbb{R}^3)$ we define the $N$-particle operator 
\begin{align}
\label{eq:def-Oi}
O^{(j)} = \mathds{1} \otimes \cdots \otimes \mathds{1} \otimes O \otimes \mathds{1} \otimes \cdots \otimes \mathds{1}
\end{align}
 as the operator acting as $O$ on  the $j$-th particle and as identity elsewhere. We consider $O^{(j)}$ as a random variable with probability distribution determined by $\psi_N $ and given through
\begin{align}
\mathbb{P}_{\psi_{N}} \left[ O^{(j)} \in A \right] = \langle \psi_{N}, \chi_A \left( O^{(j)} \right) \psi_N \rangle
\end{align}
where $\chi_A$ denotes the characteristic function of the set $A \subset \mathbb{R}$.  Since the expectation value with respect to factorized states $\psi_N = \ph^{\otimes N}$ is 
\begin{align}
\bE_{\varphi^{\otimes N}} \left[ O^{(j)} \right]  = \langle \ph, \;  O \ph \rangle  \quad \text{for all} \quad j = 1, \dots ,N , 
\end{align}
the random variables are i.i.d.  and thus,  in this case, they satisfy a law of large numbers, i.e.  for the averaged sum
$O_N = N^{-1} \sum_{j=1}^N \left( O^{(j)} - \langle \varphi, O \varphi \rangle \right) $, 
we have for any $\delta >0$
\begin{align}
\bP_{\varphi^{\otimes N}} \left[ \; \vert O_N  \vert   > \delta  \right] \rightarrow 0 \quad \text{as} \quad N \rightarrow \infty \; .
\end{align}
The large deviation principle goes one step further and investigates the rate of convergence through the rate function given by 
\begin{align}
\label{eq:rate function}
  \Lambda^*_{\psi_{N}} (x)  := - \lim_{N \rightarrow \infty} N^{-1} \log  \bP_{\psi_N} \left[  O_N  > x \right] \; ,
\end{align}
assuming the limit exists.   For i.i.d. random variables,  i.e. $\psi_N=\varphi^{\otimes N}$, Cram\'er's Theorem \cite{Cramer} shows that 
the rate function is given by
\begin{align}
\label{eq:rate-lfto}
\Lambda_{\varphi^{\otimes N}}^* ( x) = \inf_{\lambda \in \bR} \left[ - \lambda x + \Lambda_{\varphi^{\otimes N}} (\lambda) \right] 
\end{align}
where the rate function's Legendre-Fenchel transform $\Lambda_{\varphi^{\otimes N}}$ is the logarithmic moment generating function 
\begin{align}
\label{eq:rate-lft}
\Lambda_{\varphi^{\otimes N}} ( \lambda) =  \log \langle \ph, \; e^{\lambda  \left( O^{(1)} - \langle \varphi,  O \varphi\rangle \right)  }\ph \rangle . 
\end{align}

Recall that we consider the time evolution of factorized initial data with respect to \eqref{eq:Schroe}. Thus, initially the random variables are i.i.d. and therefore a law of large numbers and a large deviation principle with rate function \eqref{eq:rate-lft} hold true.  Although for times $ t >0$ the random variables are not i.i.d.\  anymore (as the factorization is not preserved),  the condensation property  \eqref{eq:BEC} ensures the validity of a law of large numbers \cite{BKS}, i.e.  for any $\delta>0$
\begin{align}
\bP_{\psi_{N,t}} \left[ \vert O_N \vert > \delta \right] \rightarrow 0 \quad \text{as} \quad N \rightarrow \infty \; .
\end{align}
In the following theorem, we show that for $t >0$  large deviation estimates hold true as well. 

Before stating our main theorem, let us introduce some notation. For $O$ a bounded self-adjoint operator on $L^2(\bR^3)$, we define the norm
\begin{align}
\label{def:norm-O}
\vertiii{O} = \| \left( - \Delta + 1 \right) O \left(  - \Delta + 1 \right)^{-1} \| 
\end{align} 
where $\|\cdot\|$ denotes the usual operator norm.  Moreover,  for $0\leq s \leq t$, let $f_{s;t} \in L^2 ( \bR^3 )$  denote the solution to
\begin{align}
\label{eq:def-fst}
i\partial_s f_{s;t} = \left( h_{\rm{H}}(s) + \widetilde{K}_{1,s} - \widetilde{K}_{2,s}  J  \right) f_{s;t}
\end{align}
with initial datum $f_{t;t} = q_t O  \ph_t = O \ph_t -\langle \ph_t, \;  O \ph_t \rangle \ph_t$, where $q_s = 1- \vert \varphi_s \rangle \langle \varphi_s \vert$, 
 $J$ denotes the anti-linear operator $Jf = \overline{f}$, the Hartree Hamiltonian $h_{\rm H}(s)$ is defined in \eqref{eq:hartree},  
and 
\begin{align}
\label{def:K}
\widetilde{K}_{1,s}= q_s  K_{1,s} \; q_s, \quad \widetilde{K}_{2,s} = q_s  K_{2,s} \; q_s 
\end{align}
with $K_{j,s}$ the operators defined by the integral kernels
\begin{align}
\label{def:K-ohne}
K_{1,s} (x,y) = v(x-y) \ph_s (x) \overline{\ph_s (y)}, \quad K_{2,s} (x,y) = v(x-y)\ph_s (x) \ph_s (y)\; .
\end{align}

\begin{theorem}
\label{thm:cramer} Assume that the interaction potential $v$ satisfies \eqref{eq:ass-v} and $\ph \in H^4 ( \bR^3 )$ with $\| \ph \|_{2} =1$.  For $t >0$, let $\psi_{N,t}$ denote the solution of the Schr\"odinger equation \eqref{eq:Schroe} with initial datum $\psi_{N,0} = \ph^{\otimes N}$ and $\varphi_t$ the solution to the Hartree equation \eqref{eq:hartree} with  $\varphi_0 = \varphi$. 

Let $O$ be a self-adjoint operator on $L^2 \left( \bR^3 \right) $ such that
$\vertiii{O}  < \infty$, and let $f_{s;t}$ be as defined above.
With $O^{(j)}$ from \eqref{eq:def-Oi}, we define  $O_{N,t} = N^{-1} \sum_{j=1}^N \left( O^{(j)} - \langle \varphi_t, O \varphi_t \rangle \right)$. 
There exist $C_1,C_2>0$ (independent of $O$) such that 
\begin{itemize}
\item[\rm{(i)}] for all  $t \geq 0$ and  $0\leq x   \leq  e^{- e^{C_1  t}} \| f_{0;t}\|_2^2/ \vertiii{O} $  
\begin{align}
\label{eq:upper-Cramer}
\limsup_{N \rightarrow \infty } N^{-1}\log  \mathbb{P}_{\psi_{N,t}} \left[ O_{N,t} > x\right] \leq -  \frac{x^2}{2 \| f_{0;t}\|_2^2}  + x^3 \frac{C_1 e^{e^{C_1  t}}\vertiii{O}^3}{\| f_{0;t}\|_2^6} . 
\end{align} 

\item[\rm{(ii)}] for all $t \geq 0$ and 
$0\leq  x \leq e^{-e^{ C_2 t} } \| f_{0;t}\|_{2}^4 / ( C_2 \vertiii{O}^3)$ 
\begin{align}
\label{eq:lower-Cramer}
\liminf_{N\rightarrow \infty} N^{-1} \log  \bP_{\psi_N}  \left[ O_{N,t} > x \right]  
\geq-  \frac{x^2}{2 \| f_{0;t}\|_2^2}  -  x^{5/2} \frac{C_2 e^{e^{C_2 t}}  \vertiii{O}^{3/2}}{\| f_{0;t}\|_{2}^4}   .
\end{align} 
\end{itemize}
\end{theorem}

We remark that the function $f_{s;t}$ is determined through Bogoliubov's quasi-free approximation of the fluctuations around the condensate (see \eqref{eq:Winfty} below).  For a detailed explanation see \cite[Theorem 2.2 and subsequent Remark]{BKS}. In fact, with the notation of \cite{BKS}, $f_{s;t} = q_s (U(t;s) + J V(t;s)) O \varphi_t$. 

The bounds  \eqref{eq:upper-Cramer} and \eqref{eq:lower-Cramer} show that the rate function of the system is, if it exists, for  sufficiently small $x>0$ given by 
\begin{align}
\label{eq:rate-function-thm}
\Lambda^*_{\psi_{N,t}}(x) = - \frac{x^2}{2 \| f_{0;t}\|_2^2} + O(x^{5/2}) .
\end{align}
In particular, Theorem \ref{thm:cramer} determines the rate function $\Lambda^*_{\psi_{N,t}}$ up to quadratic order.  Note that  for time $t=0$ the quadratic term in \eqref{eq:rate-function-thm} agrees with the one of Cram\'er's theorem \eqref{eq:rate-lft} as 
\begin{align}
\| f_{0;0} \|_2^2 = \| q_0 O \varphi \|_2^2 = \langle \varphi, \, O^2 \varphi \rangle - \vert \langle \varphi, \, O \varphi \rangle \vert^2 . 
\end{align}
In the regime $x= O(1/\sqrt{N})$, our findings agree with the central limit theorems previously obtained in \cite{BKS,BSS} proving that 
\begin{align}
\lim_{N \rightarrow \infty}\mathbb{P}_{\psi_{N,t}} \left[ \sqrt{N} O_{N,t} < x \right] = \frac{1}{\sqrt{2\pi} \| f_{0;t}\|_2} \int_{- \infty}^x e^{-\frac{r^2}{2 \| f_{0;t}\|_2^2}} dr  \; .
\end{align}
We remark that a central limit theorem still holds true when replacing the weak mean field potential (given by $N^{3\beta} v(N^\beta x )$ for $\beta =0$) with more singular interactions in the intermediate regime (corresponding to $0<\beta<1$) \cite{R}. In the physically most relevant Gross--Pitaevski regime ($\beta =1$), a central limit theorem holds for the ground state \cite{RS} showing, in particular, that the fluctuations around the condensate are approximately quasi-free.  The validity of large deviation estimates for fluctuations around the condensate for Bose-Einstein condensates in the ground state is still an open question, however.

The proof of Theorem \ref{thm:cramer} (given in Section \ref{sec:proof-thm1}) is based on a lower and an upper bound on the logarithmic moment generating function stated in the following (and proven in Section \ref{sec:proof-thm2}). 

\begin{theorem}
\label{thm:main}
Under the same assumptions as in Theorem \ref{thm:cramer}, 

\begin{itemize}
\item[\rm{(i')}] there exists a constant $C_1>0$ such that for all $0\leq \lambda \leq e^{-e^{C_1  t}} / \vertiii{O}  $  we have 
\begin{align}
\label{eq:thm-upper-bound}
\limsup_{N\to \infty} N^{-1} \log \bE_{\psi_{N,t}} e^{\lambda N O_{N,t} } \leq \frac{\lambda^2}{2} \| f_{0;t}\|_2^2 + C_1 e^{e^{C_1  t}}\lambda^3 \vertiii{O}^3  . 
\end{align}
\item[\rm{(ii')}]there exists a constant $C_2 >0$ such that for all $0 \leq \lambda \leq   e^{-e^{C_2  t}}  / \vertiii{O}   $ we have 
\begin{align}
\label{eq:thm-lower-bound}
\liminf_{N\to \infty} N^{-1} \log \bE_{\psi_{N,t}} e^{\lambda N O_{N,t} } \geq \frac{\lambda^2}{2}  \| f_{0;t}\|_2^2 - C_2e^{e^{C_2  t}} \lambda^3 \vertiii{O}^3  . 
\end{align}
\end{itemize}
\end{theorem}

The upper bound (i')  on the logarithmic moment generating function (as well as the resulting upper bound on the rate function in Theorem~\ref{thm:cramer}(i)) is an extension of the large deviation estimate obtained in \cite{KRS} to more general interaction potentials. In particular, the assumptions on the potential in Theorem \ref{thm:cramer} involve the physical interesting Coulomb potential which was excluded by the assumptions $v \in L^1( \mathbb{R}^3) \cap L^\infty ( \mathbb{R}^3)$ in \cite{KRS}. Note that the term cubic in $\lambda$ in \eqref{eq:thm-upper-bound}  depends on time through a double exponential, compared to a term exponential in time in \cite{KRS}, which is a consequence of allowing less regular interaction potentials here (entering  the proof through Lemmas~\ref{lemma:Kj} and~\ref{lemma:f}).  The quadratic term of \eqref{eq:thm-upper-bound} agrees with the findings from \cite[Theorem 1.1]{KRS}. In particular, the definition of $f_{0;t}$ in \eqref{eq:def-fst} here is the same as in \cite[ Eq.~(1.1)]{KRS}.\footnote{There are minor mistakes in \cite[Eq.~(1.1)]{KRS}, where \eqref{def:K-ohne} instead of the projected kernels \eqref{def:K} is used. This does not change the rest of the proof, which is based on estimates of quadratic forms in the truncated Fock space where the projections $q_s$ act as identities (see \eqref{eq:lastline}). Also there are two typos in \cite[Eq.~(1.1)]{KRS}, the  operators $J$ and  $K_{2,s}$ are switched and the sign in front of the operator $K_{2,s}$ is false. } 

In contrast to \cite{KRS}, we prove here also a matching lower bound (ii') on the logarithmic moment generating function (resulting, together with the upper bound, in the lower bound on  the rate function in Theorem~\ref{thm:cramer}(ii)).  This allows to determine the rate function $\Lambda_{\psi_{N,t}}^*$ up to quadratic order. In particular, we show that $\Lambda_{\psi_{N,t}}^*$ coincides up to quadratic order with the rate function of Bogoliubov's quasi-free approximation of the fluctuations around the condensate.  Whether this holds true for higher order terms  remains  an open question. In fact, we don't expect that the rate function of Bogoliubov's approximation of the fluctuations agrees with $\Lambda^*_{\psi_{N,t}}$ to all orders.

\section{Proof of Theorem  \ref{thm:cramer} } 
\label{sec:proof-thm1}

The proof of Theorem \ref{thm:cramer} uses ideas of the proof of Cramer's theorem involving estimates on the logarithmic moment generating function in Theorem \ref{thm:main}. The upper bound (i) follows by Chebychev's inequality from (i'), while the proof of the lower bound is more involved and uses both (i') and (ii'). 

\begin{proof}
\textbf{Upper bound \rm{(i)}:} For $\lambda >0$, we have 
\begin{align}
\bP_{\psi_{N,t}} \left[ O_{N,t} > x \right] = \bP_{\psi_{N,t}} \left[ e^{-\lambda N x}e^{   \lambda N   O_{N,t} } > 1 \right] .
\end{align} 
Chebychev's inequality implies that
\begin{align}
\bP_{\psi_{N,t}} \left[ O_{N,t}> x \right]  \leq  e^{-\lambda N x}\; \bE_{\psi_{N,t}} \left[ e^{  \lambda N  O_{N,t} } \right] ,
\end{align}
and we find with Theorem \ref{thm:cramer} for $\lambda < e^{-e^{C_1 t}} / \vertiii{O}$
\begin{align}
\limsup_{N \rightarrow \infty} N^{-1} \log \bP_{\psi_{N,t}} \left[ O_{N,t}> x \right] \leq -\lambda x + \frac{\lambda^2}{2} \| f_{0;t}\|_2^2 + C_1 e^{e^{C_1  t}}\lambda^3 \vertiii{O}^3 \; .
\end{align}
For 
\begin{align}
\label{eq:const0}
x < e^{-e^{C_1 t}} \| f_{0;t}\|_2^2 / \vertiii{O}
\end{align}
let $\lambda  = x/ \| f_{0;t}\|_2^2$. Then 
\begin{align}
\limsup_{N \rightarrow \infty} N^{-1} \log \bP_{\psi_{N,t}} \left[ O_{N,t}> x \right] \leq -  \frac{x^2}{2 \| f_{0;t}\|_2^2}  +\frac{x^3 C_1 e^{e^{C_1  t}}\vertiii{O}^3}{\| f_{0;t}\|_2^6} \; .
\end{align}

\textbf{Lower Bound \rm(ii):} For arbitrary $\varepsilon >0$, we have 
\begin{align}
N^{-1} \log  \bP_{\psi_{N,t}} \left[  O_{N,t} > x\right] \geq N^{-1}\log  \bP_{\psi_{N,t}} \left[ O_{N,t} \in ( x ,  x+\varepsilon)\right]  
\end{align}
and it suffices to consider in the following 
\begin{align}
\bP_{\psi_{N,t}} \left[ O_{N,t} \in ( x, x+\varepsilon)\right] &= \langle \psi_{N,t}, \chi_{(x, x+ \varepsilon)} \left(  O_{N,t} \right) \psi_{N,t} \rangle\notag\\
&= \langle \psi_{N,t}, \chi_{(x, x+\varepsilon)} \left(  O_{N,t} \right) e^{- \lambda N  O_{N,t} } e^{\lambda  N O_{N,t}} \psi_{N,t} \rangle .
\end{align}
On the support of $\chi_{(x, x+ \varepsilon)} (O_{N,t})$ we have  $e^{-\lambda N  O_{N,t}}\geq e^{-  \left( x + \varepsilon \right)  \lambda N  }$ for $\lambda >0$ and we find 
\begin{equation}
\label{eq:ct1}
 \bP_{\psi_{N,t}} \left[ O_{N,t} \in (x, x+ \varepsilon)\right] \geq e^{-  \left(  x + \varepsilon \right)  \lambda N }\langle \psi_{N,t}, \chi_{(x, x+\varepsilon)} \left(O_{N,t} \right)  e^{\lambda N  O_{N,t}} \psi_{N,t} \rangle .
\end{equation}
It is easy to check that 
\begin{equation}
\label{eq:tildeP}
\widetilde{\bP}_{\psi_{N,t}} \left[O_{N,t} \in A \right] = e^{-N \Lambda_{N,t} ( \lambda ) } \langle \psi_{N,t}, \chi_{A} \left( O_{N,t}  \right)   e^{\lambda  N O_{N,t} } \psi_{N,t} \rangle  
\end{equation}
for $A\subset \bR$ and 
\begin{align}
\Lambda_{N,t} ( \lambda ) =  N^{-1}\log \langle \psi_{N,t}, e^{\lambda N  O_{N,t}  }  \psi_{N,t} \rangle 
\end{align}
defines a probability distribution.  We use \eqref{eq:tildeP} to rewrite the expression \eqref{eq:ct1} as 
\begin{align}
 \bP_{\psi_{N,t}} \left[ O_{N,t} \in (x, x+ \varepsilon)\right] &\geq  e^{  \lambda N \left(-  x - \varepsilon\right)   + N\Lambda_{N,t}(\lambda) } \widetilde{\bP}_{\psi_{N,t}} \left[ O_{N,t}\in (x,  x+ \varepsilon) \right] \notag\\
 &= e^{  \lambda N \left(-  x - \varepsilon \right)  +N  \Lambda_{N,t}(\lambda)  }  \left( 1 -\widetilde{\bP}_{\psi_{N,t}} \left[  O_{N,t}  \leq x \right] - \widetilde{\bP}_{\psi_{N,t}} \left[  O_{N,t} \geq x + \varepsilon \right] \right)   \label{eq:lb}.
\end{align} 
Similarly to the upper bound's proof, we use Chebychev's inequality for the last two terms on the r.h.s.  and obtain for arbitrary $\lambda, \mu, \widetilde{\mu} \geq 0$ 
\begin{align}
&  \bP_{\psi_{N,t}}  \left[ O_{N,t} \in ((x, x+ \varepsilon)\right] \notag\\ 
 & \geq e^{  \lambda N \left(-  x - \varepsilon \right)  +N  \Lambda_{N,t}(\lambda) }  \left( 1 -e^{N \left(-  \Lambda_{N,t} (\lambda) +  \widetilde{\mu} x + \Lambda_{N,t} (\lambda - \widetilde{\mu}) \right) }-e^{N \left( - \Lambda_{N,t} (\lambda) -  \mu (x + \varepsilon ) +\Lambda_{N,t} (\lambda + \mu )\right) } \right)   \,.\label{eq:lb2}
\end{align} 
For given $x \in \bR$, we need to choose $\lambda,\mu, \widetilde{\mu} $ and $\varepsilon$ such that both 
\begin{align}
\limsup_{N\to \infty} \left( -   \Lambda_{N,t} (\lambda) +  \widetilde{\mu} x  + \Lambda_{N,t} (\lambda - \widetilde{\mu}) \right)<0
\end{align}
and 
\begin{align}
\limsup_{N\to\infty} \left(  -  \Lambda_{N,t} (\lambda) -  \mu (x + \varepsilon ) +\Lambda_{N,t} (\lambda + \mu ) \right)<0\,.
\end{align}
In fact, for $0 < \delta < \widetilde{\delta}$, let $\lambda = x ( 1 + \delta ) / \| f_{0;t} \|_2^2, \varepsilon = x \widetilde{\delta}, \widetilde{\mu} = \delta x / \| f_{0;t}\|^2_{2}$ and $\mu = ( \widetilde{\delta} - \delta ) x / \| f_{0;t}\|_{2}^2$.  Then, as long as, 
\begin{align}
\label{eq:const1}
0\leq x \leq \frac{\min \lbrace  e^{-e^{C_1 t}}, \;   e^{-e^{C_2 t}}/(1 + \delta)\rbrace \| f_{0;t}\|_{2}^2} { \vertiii{O} }
\end{align}
we have with Theorem \ref{thm:main}
\begin{align}
&\limsup_{N\to \infty} \left( -\Lambda_{N,t}(\lambda) + \widetilde \mu x + \Lambda_{N,t}(\lambda -\widetilde \mu) \right)  \nonumber \\ & \leq  -  \frac{\lambda^2}{2} \| f_{0;t}\|_{2}^2 +  C_2 e^{e^{C_2 t}} \lambda^3 \vertiii{O}^3 + \widetilde \mu x +  \frac{(\lambda-\widetilde \mu)^2}{2}\| f_{0;t}\|_{2}^2 + C_1 e^{e^{C_1 t}} (\lambda-\widetilde \mu)^3 \vertiii{O}^3\notag
\\ & =  -  \frac{ x^2 \delta^2}{2 \| f_{0;t}\|_{2}^2} +   \frac{x^3   \vertiii{O}^3}{\|f_{0;t}\|^6}  \left(  C_2 e^{e^{C_2 t}}(1+\delta)^3    + C_1 e^{e^{C_1 t}}\right) <0
\end{align}
if
\begin{align}\label{eq:const2}
  x  \vertiii{O}^3 \left(  C_2 e^{e^{C_2 t}} (1+\delta)^3    + C_1 e^{e^{C_1 t}}\right) <   \frac{  \delta^2}{2} \| f_{0;t}\|_{2}^4\,.
\end{align}
Similarly, as long as 
\begin{align}\label{eq:const3}
0\leq x \leq \frac{ \min\lbrace e^{-e^{C_1 t}}/(1+\widetilde\delta),e^{-e^{C_2 t}}/(1+\delta) \rbrace \| f_{0;t}\|_{2}^2  } {\vertiii{O}}
\end{align}
we have from Theorem \ref{thm:main}
\begin{align}
&\limsup_{N\to \infty} \left( -\Lambda_{N,t}(\lambda)  - \mu (x + \varepsilon)  + \Lambda_{N,t}(\lambda + \mu) \right) \notag\\
 & \leq    -  \frac{\lambda^2}{2}\| f_{0;t}\|_{2}^2 +   C_2 e^{e^{C_2 t}}\lambda^3 \vertiii{O}^3- \mu (x + \varepsilon)  +  \frac{(\lambda + \mu)^2}{2} \| f_{0;t}\|_{2}^2 + C_1 e^{e^{C_1 t}}  (\lambda + \mu)^3 \vertiii{O}^3 \notag
\\ & =  -  \frac{ x^2 (\widetilde \delta -\delta)^2}{2 \| f_{0;t}\|_{2}^2 } +   \frac{x^3  \vertiii{O}^3}{\| f_{0;t}\|_{2}^6}  \left(  C_2 e^{e^{C_2 t}}(1+\delta)^3    + C_1 e^{e^{C_1 t}} (1+\widetilde \delta)^3\right) <0
\end{align}
if
\begin{equation}\label{eq:const4}
  x  \vertiii{O}^3  \left( C_2 e^{e^{C_2 t}} (1+\delta)^3    + C_1 e^{e^{C_1 t}}(1+\widetilde \delta)^3\right) <   \frac{  (\widetilde \delta - \delta)^2}{2}  \| f_{0;t}\|_{2}^4\,.
\end{equation}
In particular, under conditions \eqref{eq:const1}, \eqref{eq:const2}, \eqref{eq:const3} and \eqref{eq:const4}, we have from \eqref{eq:lb2}
\begin{align}
\liminf_{N \rightarrow \infty} N^{-1} \log \bP \left[ O_{N,t} >x\right]   &\geq \liminf_{N \rightarrow \infty} \Lambda_{N,t}(\lambda) -\lambda(x+\varepsilon) \notag\\ &\geq - \frac{x^2 }{2 \| f_{0;t}\|_{2}^2}(1+ 2 \widetilde\delta(1+\delta) )   - C_2 e^{e^{C_2 t}}  x^3 (1+\delta)^3 \frac{ \vertiii{O}^3}{ \| f_{0;t}\|_{2}^6}\,.
\end{align}
With $C_3=\max\{C_1,C_2\}$ 
we can take  
\begin{equation}
\delta^2 = \left( \frac{\tilde\delta}{2}\right)^2 = 70\, C_3 e^{e^{C_3 t} }  \frac { x \vertiii{O}^3} { \| f_{0;t}\|_{2}^4}
\end{equation}
and \eqref{eq:const2} as well as \eqref{eq:const4} are satisfied as long as $\delta <1$. 
For
\begin{align}
x \leq \min \left\lbrace \tfrac 13 e^{-e^{C_3 t}} \| f_{0;t}\|_{2}^2 / \vertiii{O}, \;  \tfrac 1{70\, C_3} e^{- e^{ C_3 t} } \| f_{0;t}\|_{2}^4 /  \vertiii{O}^3 \right\rbrace 
\end{align}
we can thus  conclude that
\begin{align}
\liminf_{N \rightarrow \infty} N^{-1} \log \bP \left[ O_{N,t}>x\right]   \geq - \frac{x^2 }{2  \| f_{0;t}\|_{2}^2} - C_4 e^{e^{C_4 t}}  x^{5/2} \frac{\vertiii{O}^{3/2}}{\| f_{0;t}\|_{2}^4}  
\end{align}
for $C_4 >0$ large enough. We shall show in Lemma~\ref{lemma:f} below that $\|f_{0;t}\|_2 \leq \vertiii{O} e^{C|t|}$ for suitable $C>0$, which allows for the simpler condition on $x$ as stated in Theorem~\ref{thm:cramer}(ii). 
\end{proof}

\section{Proof of Theorem \ref{thm:main}} 
\label{sec:proof-thm2}

\subsection{Properties of $K_{j,s}$ and $f_{s;t}$}
\label{sec:prel-K}
 
In this section, we show in Lemma \ref{lemma:f} useful estimates on the function $f_{s;t}$ defined in \eqref{eq:def-fst}.  To this end, we first collect in Lemma~\ref{lemma:Kj} properties of the kernels $K_{j,s}$ defined in \eqref{def:K-ohne}.  
 These Lemmas are the crucial ingredient to generalize the result of \cite{KRS} to more singular interaction potentials.  The main difference is that in \cite{KRS} estimates of the form \eqref{eq:estimates_K-2} rely on the propagation of the $H^1$-norm of $\varphi_s$ using, in particular, that by conservation of energy and \eqref{eq:ass-v} we have
\begin{align}
\label{eq:estimate-phi-H1}
\| \varphi_s \|_{H^1(\bR^3)} \leq C \| \varphi \|_{H^1(\bR^3)} 
\end{align}
for a constant $C>0$.  In contrast, here, we need the propagation of higher Sobolev norms of $\varphi_s$ in \eqref{eq:estimate-K-crucial-step}, i.e. bounds of the form 
\begin{align}
\label{eq:estimate-phi-H^2}
\| \varphi_s \|_{H^k (\bR^3)} \leq C e^{C s}\| \varphi_0 \|_{H^k (\bR^3)} 
\end{align}
for $k \geq 2$ which are well-known (see e.g. \cite{Caz}).  These lead to bounds exponential in time in  \eqref{eq:estimates_K-2}) and, thus, to bounds double exponential in time in  \eqref{eq:estimate-f-H2}) because of the use of a Gronwall type estimate. These bounds effect Lemma \ref{lemma:step2} and, consequently, the error terms in Theorems~\ref{thm:cramer} and ~\ref{thm:main}. 

\begin{lemma}
\label{lemma:Kj}
For $s \in \bR$ and $v$ satisfying \eqref{eq:ass-v},  let $\varphi_s$ denote the solution to the Hartree equation \eqref{eq:hartree} with initial data $\varphi \in H^4( \bR^3)$. There exists a constant $C>0$ such that 
\begin{align}
\label{eq:estimate-v1}
\| v * \vert \ph_s \vert^2& \|_{\infty} \leq C \; , \\
\|  \nabla v * \vert \ph_s \vert^2 \|_{\infty} \leq C\, e^{C \vert s \vert } \ , &\quad\|  \Delta v * \vert \ph_s \vert^2 \|_{\infty} \leq C e^{C \vert s \vert } \; , \label{eq:estimate-v2}
\end{align}
and, furthermore for $j=1,2$ and $f \in H^2 ( \bR^3)$ 
\begin{align}
\label{eq:estimates_K-1}
\| K_{j,s} & \|_{L^2 \left( \bR^3 \times \bR^3 \right)} \leq C  \\
\| \nabla K_{j,s} f \|_2   \leq C e^{C \vert s \vert} \| f \|_{H^1 (\bR^3)}, & \quad \| \Delta K_{j,s} f \|_2   \leq C e^{C \vert s \vert} \| f \|_{H^2 (\bR^3)} .\label{eq:estimates_K-2}
\end{align}
\end{lemma} 

\begin{proof}  From \eqref{eq:ass-v}, we have
\begin{align}\label{asy}
\vert \left( v* \vert \varphi_s \vert^2 \right) (x) \vert  \leq  \int \vert v(x-y) \vert^2 \;  \vert \varphi_s (y) \vert^2 dy+ C \| \varphi_s \|_{2}^2 \leq C \| \varphi_s \|_{H^1( \bR^3)}^2 
\end{align}
and \eqref{eq:estimate-v1} follows from \eqref{eq:estimate-phi-H1}. 
Similarly, since 
\begin{align}
\nabla v* \vert \varphi_s \vert^2 = 2 v* \Re\, \overline{\varphi_s} \nabla \varphi_s 
\end{align}
we have with \eqref{eq:estimate-phi-H^2}
\begin{align}
\|  \nabla v * \vert \ph_s \vert^2 \|_{\infty} \leq C \| \varphi_s \|_{H^2 (\bR^3)}^2 \leq C e^{C \vert s \vert } \, .
\end{align}
The second bound in \eqref{eq:estimate-v2}  follows in the same way. 

Moreover, with \eqref{def:K-ohne}, we have for $j=1,2$
\begin{align}
\| K_{j,s} \|_{L^2 \left( \bR^3 \times \bR^3 \right)}^2 = \int \vert \ph_s (x) \vert^2 v^2(x-y)\vert \ph_s (y)\vert^2 \; dxdy = \langle \ph_s,  \;\left( v^2 * \vert \ph_s \vert^2 \right)\ph_s \rangle 
\label{eq:K1}
\end{align}
and thus  \eqref{eq:estimates_K-1} follows by arguing as in  \eqref{asy} above.  In order to show \eqref{eq:estimates_K-2}, we integrate by parts 
\begin{align}
\int  \left( \nabla_x K_{1,s} \right) (x,y)  f  (y ) dy  =& \int v(x-y) \left( \nabla \varphi_s\right) (x ) \overline{\varphi_s (y)} f(y) dy \notag\\
&+  \int v(x-y) \varphi_s (x) \left( \overline{\nabla \varphi_s}\right)  (y) f (y) dy \notag\\
&+ \int v(x-y) \varphi_s (x)  \overline{ \varphi_s (y)} \nabla f (y) dy\label{eq:estimate-K-crucial-step}
\end{align}
and estimate with \eqref{eq:ass-v} similarly as above 
\begin{align}
\label{eq:eq2}
\| \nabla K_{1,s} f \|_2 \leq C \| \varphi_s \|_{H^3( \bR^3)}^2 \| f \|_{2}+  C \| \varphi_s \|_{H^2 ( \bR^3)}^2 \| f \|_{H^1 (\bR^3)}  \leq C e^{C \vert s \vert} \| f \|_{H^1 (\bR^3)} \; .
\end{align}
The second estimate in \eqref{eq:estimates_K-2} follows in the same way. 
\end{proof}

Because of \eqref{eq:estimate-phi-H^2}, one readily checks that the same bounds hold with $\widetilde{K}_{j,s}$ in place of $K_{j,s}$. Those bounds are in fact the ones we need below.

\begin{lemma}
\label{lemma:f} Under the same assumptions as in Theorem \ref{thm:main}, let $f_{s;t}$ be defined as in \eqref{eq:def-fst}.  Then, there exists a constant $C>0$ such that for all  $0\leq s \leq t$ we have
\begin{align}
\| f_{s;t} \|_2 \leq  \| O \| e^{C \vert t -s \vert } , \quad \| f_{s;t} \|_{H^2 \left( \bR^3 \right)} \leq  Ce ^{e^{C \vert t \vert }}\vertiii{O} \label{eq:estimate-f-H2} 
\end{align}
where $\vertiii{O}$ is  defined in \eqref{def:norm-O}. 
\end{lemma}

\begin{proof}
Since 
\begin{align}
\| f_{s;t} \|_2^2 =& \| f_{t;t} \|_2^2 - \int_s^t  \partial_\tau \| f_{\tau;t} \|_2^2  \; d\tau =  \| f_{t;t} \|_2^2 - 2 \int_s^t   \Im \langle f_{\tau;t},   \widetilde{K}_{2,\tau} J f_{\tau;t} \rangle \; d\tau
\end{align}
we have with \eqref{eq:estimates_K-1} 
\begin{align}
\| f_{s;t} \|_2^2 \leq  \| f_{t;t} \|_2^2 + C \int_s^t  \| f_{\tau;t} \|_2^2 \; d\tau \; . 
\end{align}
Since $\| f_{t;t} \|_2 = \| q_t O \varphi_t \|_2 \leq  \| O \| $, the first bound in \eqref{eq:estimate-f-H2}  is a consequence of Gronwall's inequality. 

In order to show the second, we compute 
\begin{align}
\label{eq:H2-estimate1}
\partial_s  \|  f_{s;t} \|_{H^2 \left( \bR^3 \right)}^2 & = 2\, \Im \, \langle \left(   \Delta^2 - 2 \Delta \right) f_{s;t}, \; \left(  v* \vert \varphi_s \vert^2 +  \widetilde{K}_{1,s} \right) f_{s;t}\rangle  - 2\, \Im \, \langle  \left(- \Delta + 1 \right)^2 f_{s;t}, \;   \widetilde{K}_{2,s} J f_{s;t}\rangle\notag\\
  & = 4\, \Im\, \langle  \Delta  f_{s;t}, \;    \left( \nabla v* \vert \varphi_s \vert^2  \right) \nabla f_{s,t} \rangle + 2\, \Im\, \langle\Delta f_{s;t},     \left( \Delta v* \vert \varphi_s \vert^2  \right)  f_{s,t} \rangle \notag \\ &\quad + 2\, \Im\, \langle \Delta f_{s;t}, \;    \Delta \widetilde{K}_{1,s}  f_{s,t}\rangle  - 4\,\Im\, \langle \Delta f_{s;t}, \;    \left( v* \vert \varphi_s \vert^2 + \widetilde{K}_{1,s} \right)  f_{s,t} \rangle  \notag\\ & \quad - 2\, \Im\, \langle \left( - \Delta + 1 \right) f_{s;t}, \;   \left(- \Delta +1 \right)   \widetilde{K}_{2,s} J f_{s,t} \rangle \,.
\end{align} 
It follows from Lemma~\ref{lemma:Kj} that all the terms on the r.h.s. can be bounded by $C\|f_{s;t}\|_{H^2(\bR^3)}^2 e^{Cs}$. The second bound in \eqref{eq:estimate-f-H2} thus also follows from Gronwall's inequality, together with
\begin{align}
\| f_{t;t} \|_{H^2 \left( \bR^3 \right)} = \| q_t O \varphi_t \|_{H^2 \left( \bR^3 \right)} \leq  \|\varphi_t\|_{H^2(\bR^3)} \|O\| + \| O \varphi_t \|_{H^2 \left( \bR^3 \right)}
\leq
 \left( \| O \| + \vertiii{O}  \right)  \| \varphi_t \|_{H^2 \left( \bR^3 \right)}
\end{align}
and \eqref{eq:estimate-phi-H^2}.
\end{proof}

Note that the generalization of the interaction potential comes into play when using the estimates \eqref{eq:bound-h1} and \eqref{eq:f-h2} from Lemma \ref{lemma:Kj} and Lemma \ref{lemma:f}. These estimates lead to the bounds double exponential in time.

\subsection{Fluctuations around the condensate} For the proof of Theorem \ref{thm:main}, we need to study the fluctuations around the condensate in the truncated Fock space of excitations. This description is based on the observation of \cite{LNSS} that any $N$-particle bosonic wave function $\psi_N \in L_{\rm s}^2( \bR^{3N})$ can be decomposed as 
\begin{align}
\psi_N = \eta_0 \;  \varphi_t^{\otimes N}  + \eta_1  \otimes_{\rm s} \varphi_t^{\otimes (N-1)} + \cdots + \eta_N 
\end{align}
with $\eta_j \in L^2_{\perp \varphi_t} \left( \bR^3 \right)^{\otimes_{\rm s} j}$, where  $L^2_{\perp \varphi_t} ( \bR^3 )$ denotes the orthogonal complement in $L^2 ( \bR^3) $ of the condensate wave function $\varphi_t$ and $\otimes_{\rm s}$ the symmetric tensor product.  In particular, this observation allows to define the unitary operator 
\begin{align}
\label{def:U}
\cU_t: L_{\rm s}^2 \left( \bR^{3N} \right) \rightarrow \cF_{\varphi_t}^{\leq N} = \bigoplus_{j=0}^N L^2_{\perp \varphi_t} \left( \bR^3 \right)^{\otimes_{\rm s} j }
\end{align}
mapping an $N$-particle bosonic wave function $\psi_N$ onto an element of the truncated Fock space, with $\cU_t \psi_N = \lbrace \eta_0, \dots, \eta_N \rbrace$ describing the excitations orthogonal to the condensate.  On the full bosonic Fock space (built over $L^2( \bR^3)$) we have the usual creation and annihilation operators, given for $f \in L ^2( \bR^3 )$ by 
\begin{align}
a^*(f) = \int  f(x) \, a^*_x \, dx, \quad a(f) = \int  \overline{f(x)} \,  a_x \, dx
\end{align}
and the number of particles operator $\cN = \int a_x^* a_x dx $.  Moreover, we have the modified creation and annihilation operators $b^*(f), b(f)$ which (in contrast to $a^*(f), a(f)$) leave the truncated Fock space $\cF_{\varphi_t}^{\leq N}$ invariant and are given for $f \in L^2_{\perp \varphi_t} ( \bR^3 ) $ by 
\begin{align}
\label{def:b}
b^*(f)  &= \cU_t \; a^*(f) \tfrac{a(\varphi_t)}{\sqrt{N}}\;  \cU_t^* = a^*(f) \sqrt{ 1- \tfrac{\cN_+(t)}{N}} \notag\\
b(f)  &= \cU_t \;  \tfrac{a^*(\varphi_t)}{\sqrt{N}} a(f) \; \cU_t^* =\sqrt{ 1- \tfrac{\cN_+(t)}{N}} a(f) 
\end{align}
where $\cN_+(t) = \cN - a^*(\varphi_t) a(\varphi_t)$ is the number of excitations.  Note that the operators $b^*(f), b(f)$ are time dependent,  yet we omit the time dependence in their notation for simplicity. Their commutators given for $f_1,f_2 \in L^2_{\perp \varphi_t} ( \bR^3 ) $ by
\begin{align}
\left[ b(f_1), b^*(f_2) \right] = \left( 1 - \frac{\mathcal{N}_+(t)}{N} \right) \langle f_1, \; f_2 \rangle - \frac{1}{N}  a^*(f_2) a(f_1) , \quad \left[ b(f_1), b(f_2) \right]  = \left[ b^*(f_1), b^*(f_2) \right] = 0 
\end{align}
behave in the limit $N \rightarrow \infty$ similarly as the standard commutation relations of $a^*(f_1), a(f_2)$; however, the correction terms of order $N^{-1}$ lead to technical difficulties in the proofs below. With \eqref{def:b} and the following further properties of $\cU_t$
\begin{align}
\label{eq:propU}
\cU_t a^*( \varphi_t) a( \varphi_t) \cU_t^* &= N - \cN_+ (t) \quad \notag\\
\cU_t a^*(f) a(g) \cU_t^* &= a^*(f) a(g) 
\end{align}
for $f,g \in L_{\perp \varphi_t }^2 ( \bR^3 )$, we can compute the generator  $\cL_N(t)$ of the fluctuation dynamics
\begin{align}
\label{def:flucdyn}
\cW_N(t_2;t_1) = \cU_{t_2} e^{-i H_N (t_2 -t_1)}\cU_{t_1}^*: \cF_{\varphi_{t_1}}^{\leq N} \rightarrow\cF_{\varphi_{t_2}}^{\leq N} \;
\end{align}
 defined by
\begin{align}
i \partial_{t_2} \cW_N (t_2; t_1) = \cL_N (t_2) \cW_N(t_2;t_1) \,.
\end{align}
For $\xi_1, \xi_2 \in \cF_{\perp \varphi_t}^{\leq N}$ it is given by 
\begin{align}
\langle \xi_1, \cL_N (t) \xi_2 \rangle & = \langle \xi_1, \left[ i \partial_t \cU_t \right] \cU_t^* + \cU_t H_N \cU_t^* \xi_2 \rangle \notag\\
 &= \langle \xi_1, d\Gamma (h_{\rm H} (t) + K_{1,t})  \xi_2 \rangle + \Re  \int  \;  K_{2,t} (x,y) \, \langle \xi_1, b_x^* b_y^* \xi_2 \rangle \; dxdy \notag \\ 
&\quad - \frac{1}{2N} \langle \xi_1 , d\Gamma (v *|\ph_{t}|^2 + K_{1,t} - \mu_{t}) (\cN_+ (t) - 1) \xi_2 \rangle \notag\\
 &\quad + \frac{2}{\sqrt{N}} \Re \, \langle \xi_1, \cN_+(t)  b ((v*|\ph_{t}|^2) \ph_{t}) \xi_2 \rangle\notag \\
  &\quad +\frac{2}{\sqrt{N}} \int  \;  v (x-y)  \Re \,  \ph_{t} (x) \langle \xi_1, a_y^*a_{x}b_{y} \xi_2 \rangle\; dx dy \notag  \\
   &\quad + \frac{1}{2N} \int \,v (x-y) \langle \xi_1 , a_x^* a_y^* a_{x} a_{y} \xi_2 \rangle \,   dx dy \label{eq:form-LN}
\end{align}
where we used the notation introduced in \eqref{eq:hartree}, \eqref{def:K-ohne} and  $2\mu_{t} = \int dx dy \; v(x-y)  \vert \varphi_{t} (x) \vert^2  \vert \varphi_{t} (y) \vert^2 $.  

In the limit of large $N$, the fluctuation dynamics $\cW_N (t_2; t_1)$ can be approximated by a limiting dynamics $\cW_\infty (t_2 ;t_1) : \cF_{\perp \ph_{t_1}}   \to \cF_{\perp \ph_{t_2}} $ which is obtained by taking a formal limit $N\to \infty$ in \eqref{eq:form-LN}. It satisfies the equation 
\begin{align}
\label{eq:Winfty} 
i\partial_t \cW_\infty (t_2 ; t_1 ) = \cL_\infty (t_2) \cW_\infty (t_2;t_1)
\end{align} 
with the generator $\cL_\infty (t)$ whose matrix elements are given for $\xi_1, \xi_2 \in \cF_{\perp \ph_{t}}$ by 
\begin{align}
\langle \xi_1 , \cL_\infty (t) \xi_2 \rangle  = \langle \xi_1, d\Gamma (h_{\rm H} (t) + K_{1,t}) \xi_2 \rangle  + \Re \int K_{2,t} (x;y) \langle \xi_1, a_x^* a_y^* \xi_2 \rangle \; dx dy 
\end{align}  
For more details see \cite{LNS} resp. \cite{C,GMM,H,MPP}. The generator $\cL_\infty (t)$ of the limiting fluctuation dynamics is quadratic in creation and annihilation operators and thus gives rise to a Bogoliubov transformation \cite{BKS,BPPS,R} related to the function $f_{0;t}$ defined in \eqref{eq:def-fst} (see \cite[Theorem 2.2 et seq.]{BKS}).

\subsection{Proof of Theorem \ref{thm:main} } \label{sec:proof-lower bound} 

The proof follows closely the ideas of \cite{KRS} and is based on Baker--Campbell--Hausdorff formulas proved therein (see \cite[Propositions~2.2--2.5]{KRS}.  The main difference compared to \cite{KRS} is, on the one hand, our weaker assumptions on the interaction potential (entering in the estimates \eqref{eq:bound-h1} and \eqref{eq:f-h2} through Lemmas~\ref{lemma:Kj} and~\ref{lemma:f}).  On the other hand, we prove lower bounds in Lemmas~\ref{lemma:step1}--\ref{lemma:step3} as well,  based on similar ideas as for the upper bounds (see also \cite[subsequent discussion of Theorem 1.1]{KRS}). 

With the map $\cU_0$ defined in \eqref{def:U}, we observe that $\ph^{\otimes N} = \cU_0^* \Omega $ and thus by definition of the fluctuation dynamics in \eqref{def:flucdyn}, we have 
\begin{align}
\psi_{N,t} = e^{-iH_N t} \ph^{\otimes N} = e^{-iH_N t} \cU_0^* \Omega = \cU_t^* \cW_N(t;0 ) \Omega . 
\end{align}
Hence we can write the 
moment generating function as 
 \begin{align}
 \bE_{\psi_{N,t}} \left[ e^{\lambda N O_{N,t} } \right] =& \left\langle \psi_{N,t}, \;e^{\lambda N O_{N,t}} \psi_{N,t} \right\rangle \notag\\
 =& \left\langle \Omega, \; \cW_N^*(t;0) \cU_t e^{\lambda d \Gamma ( \widetilde{O}_t ) } \cU_t^* \cW_N (t;0) \Omega \right\rangle  \label{eq:starto}
 \end{align}
 with $\widetilde{O}_t = O - \langle \varphi_t, O \varphi_t \rangle$.  The properties \eqref{def:b}, \eqref{eq:propU} of $\cU_t$ allow to compute
\begin{align}
\cU_t d \Gamma (\widetilde{O}_t ) \cU_t^* = d\Gamma ( q_t \widetilde{O}_t q_t ) + \sqrt{N} \phi_+ ( q_t O \varphi_t ) 
\end{align}
 where we used that $\langle \varphi_t, \widetilde{O}_t \varphi_t \rangle =0$ and we introduced the notation 
 $ \phi_+(h) = b(h) + b^*(h)$ for  $ h \in L^2_{\perp \ph_t} ( \bR^3 )$. Thus, we arrive at 
  \begin{align}
 \bE_{\psi_{N,t}} \left[ e^{\lambda N O_{N,t} } \right] 
 = \left\langle \Omega, \; \cW_N^*(t;0) e^{\lambda d \Gamma ( q_t\wO_t q_t) + \lambda \sqrt{N} \phi_+( q_t O \ph_t )} \cW_N (t;0) \Omega \right\rangle  \; . \label{eq:start}
 \end{align}
As in \cite{KRS}, we split the proof into three steps.  The first step,   Lemma \ref{lemma:step1},  can be proved as \cite[Lemma 3.1]{KRS}. 

\begin{lemma} 
\label{lemma:step1}
There exists a constant $C>0$ such that for all $t \in \bR$ and $\lambda \leq \| O\|^{-1}$
\begin{align}
&e^{-C N \| O \|^3 \lambda^3} \left\langle \Omega, \; \cW_N^*(t;0) e^{\lambda \sqrt{N} \phi_+ (q_t O \ph_t )/2 } e^{-2 \lambda \| O \| \cN_+(t) } e^{\lambda \sqrt{N} \phi_+ (q_t O \ph_t )/2 } \cW_N(t;0) \Omega \right\rangle \nonumber \\
& \leq \left\langle \cW_N^* (t;0) \Omega, \; e^{\lambda d \Gamma ( q_t \wO_t q_t ) + \lambda \sqrt{N} \phi_+ (q_t O \ph_t)} \cW_N(t;0) \Omega\right \rangle \notag\\
&  \leq e^{C N \| O \|^3 \lambda^3} \left\langle \Omega, \; \cW_N^*(t;0) e^{\lambda \sqrt{N} \phi_+ (q_t O \ph_t )/2 } e^{2 \lambda \| O \| \cN_+(t) } e^{\lambda \sqrt{N} \phi_+ (q_t O \ph_t )/2 } \cW_N(t;0) \Omega \right\rangle \,. \label{eq:lemma-step1-upper}
\end{align}
\end{lemma}
 
 \begin{proof}
The proof of the upper bound in \eqref{eq:lemma-step1-upper} is the same as  in \cite[Lemma 3.1]{KRS}.  The lower bound can be proved in essentially the same way. For completeness we carry it out in the following. 
 As in \cite{KRS} (but replacing $\kappa$ with $- \kappa$), we define for $s \in [0,1]$ and  $\kappa> 0$  the vector 
\begin{align}
\label{eq:xi-1}
\xi_s = e^{-(1-s)\lambda \kappa \cN_+ (t) /2} e^{(1-s) \lambda \sqrt{N} \phi_+ (q_t O \ph_t)/2} e^{s \lambda \left[ d \Gamma (q_t \wO_t q_t ) + \sqrt{N} \phi_+(q_t O \ph_t) \right]/2} \cW_N (t;0) \Omega .
\end{align}
We have 
\begin{align}
\label{eq:xi-10}
\| \xi_0 \|^2 =\left\langle \Omega,  \; \cW_N^*(t;0) e^{\lambda \sqrt{N} \phi_+ (q_t O \ph_t )/2 } e^{-\kappa \lambda  \cN_+(t) } e^{\lambda \sqrt{N} \phi_+ (q_t O \ph_t )/2 } \cW_N(t;0) \Omega \right\rangle
\end{align}
and 
\begin{align}
\label{eq:xi-11}
\| \xi_1 \|^2 =  \left\langle \Omega, \; \cW_N^* (t;0) e^{\lambda d \Gamma ( q_t \wO_t q_t ) + \lambda \sqrt{N} \phi_+ (q_t O \ph_t)} \cW_N(t;0) \Omega\right \rangle .
\end{align}
To control the difference of \eqref{eq:xi-10} and \eqref{eq:xi-11}, we compute the derivative
\begin{align}
\label{eq:Re-1}
\partial_s \| \xi_s \|^2 = 2 \Re \langle \xi_s, \,  \partial_s \xi_s \rangle = 2 \Re \langle \xi_s, \,  \cM_s \xi_s \rangle
\end{align}
where the operator $\cM_s$ is given by
\begin{align}
\label{eq:def-M}
\cM_s =& \frac{\lambda}{2} e^{-(1-s)\lambda \kappa \cN_+ (t)/2} e^{(1-s)\lambda \sqrt{N} \phi_+ (q_t O \varphi_t )/2} d \Gamma (q_t \wO_t q_t ) e^{-(1-s)\lambda \sqrt{N} \phi_+ (q_t O \varphi_t )/2} e^{(1-s)\lambda \kappa \cN_+ (t)/2}  \notag\\
&+ \frac{\lambda \kappa}{2} \cN_+ (t) .
\end{align} 
With \cite[Propositions~2.2--2.4]{KRS}, we can compute $\cM_s$ explicitly. Note that only the hermitian part of $\cM_s$ enters in \eqref{eq:Re-1}. Using the notation $h_t = (1-s) \lambda q_t O \ph_t$  and $\gamma_s = \cosh s, \sigma_s = \sinh s$, we find
 \begin{align}
 \label{eq:M1}
\frac{\cM_s + \cM_s^*}{\lambda}  & = d \Gamma (q_t \wO_t q_t) -  \frac{\sigma_{\| h_t \|}^2}{\| h_t \|^2} \langle h_t , \wO_t  h_t \rangle \left(N - {\cN_+ (t)} \right) + \left( \frac{\gamma_{\| h_t \|} - 1}{\| h_t \|^2} \right)^2  \langle h_t , \wO_t  h_t \rangle a^* (h_t) a (h_t)\notag \\
 &\quad + \frac{\gamma_{\| h_t \|} -1}{\| h_t \|^2} (a^* (h_t) a (q_t \wO_t h_t) + a^* (q_t \wO_t h_t) a (h_t) ) + \kappa \cN_+ (t)\notag \\
   &\quad + \sqrt{N} \,  \frac{\sigma_{\| h_t \|}}{\| h_t \|}
  \sinh ((s-1) \lambda \kappa /2)  \left[ \frac{\gamma_{\| h_t \|} - 1}{\| h_t \|^2} \langle h_t , \wO_t h_t \rangle \phi_+ (h_t)  + \phi_+ (q_t \wO_t h_t) \right]   \,.
 \end{align}
 
For any $h \in L^2_{\perp \ph} ( \bR^3 )$ and any bounded operator $H$ on $L^2_{\perp \ph} ( \mathbb{R}^3 )$, we have the bounds
\begin{align}
\label{eq:bounds-b}
\| b(h) \xi \| \leq \| h \|_2 \| \cN_+(t)^{1/2} \xi \|, \quad \| b^*(h) \xi \| \leq \| h\|_2 \| \left( \cN_+(t) +1 \right)^{1/2} \xi \|,\quad \pm d\Gamma (H) \leq \| H \| \cN_+ (t) \,.
\end{align}
Consequently, all  terms on the r.h.s. of \eqref{eq:M1} can be bounded by a constant of order $N$.  Furthermore, since 
\begin{align}
\| \widetilde{O}_t \| \leq \| O \| ( 1 + \| \varphi_t \|_2^2 ) =2  \| O \|
\end{align} 
we can bound $d \Gamma (q_t \wO_t q_t) \geq - 2 \|O\| \cN_+(t)$ and hence the choice $\kappa = 2 \|O\|$ gives $d \Gamma (q_t \wO_t q_t) + \kappa \cN_+(t) \geq 0$. 
Moreover, since 
\begin{align}
\| h_t \|_2 \leq \lambda \| q_t O \varphi_t\|_2 \leq  \lambda \| O \|  \,  \|\varphi_t\|_2\leq 1
\end{align} 
for all  $\lambda \leq \| O \|^{-1}$, all the other terms on the r.h.s. of \eqref{eq:M1} are at least of order $\lambda^2$.  Thus, using \eqref{eq:bounds-b}  and  $\kappa = 2 \|O\|$  we obtain the  lower bound 
 \begin{align}
 \label{eq:lower-1}
\frac{2}{\lambda} \text{Re } \langle \xi_{s} , \cM_{s} \xi_{s} \rangle  \geq 
- C \lambda^2 N \| O \|^3   \| \xi_{s} \|^2 \,.
 \end{align}
 In combination with \eqref{eq:Re-1} the lower bound in \eqref{eq:lemma-step1-upper} now follows from 
 Gronwall's inequality. 

The proof of the upper bound in \cite{KRS} works in the same way, simply replacing $\kappa$ by $-\kappa$ and estimating the terms in \eqref{eq:M1} from above instead of from below. 
\end{proof}

The second step, Lemma \ref{lemma:step2}, is a generalization of Lemma \cite[Lemma 3.2]{KRS} to more singular interaction potentials.  The proof involves Lemmas~\ref{lemma:Kj} and~\ref{lemma:f} (see in particular \eqref{eq:bound-h1} and \eqref{eq:f-h2}) for the estimates yielding to an double exponential in time of the term cubic in $\lambda$ and in the definition of $\kappa_s$(compared to an exponential in time in  \cite{KRS}). We remark that for Lemma \ref{lemma:step2} it is a crucial observation that $f_{s;t} \in L^2_{\perp \varphi_s} ( \bR^3)$ for all $0\leq s \leq t$. This follows from the fact that $\langle \varphi_t, f_{t;t} \rangle = \langle \varphi_t , q_t O \varphi_t\rangle=0$ by construction, as well as 
\begin{align}
\partial_s \langle \varphi_s, f_{s;t} \rangle =   - i \,  \Im\, \langle \varphi_s, \left[  \widetilde{K}_{1,s} - \widetilde{ K}_{2,s} J  \right] f_{s;t }\rangle =0 
\end{align}
using the definitions \eqref{eq:def-fst} and \eqref{def:K}.

 \begin{lemma}
 \label{lemma:step2}
 For $0\leq s \leq t$,  let $f_{s;t} \in L^2_{\perp \ph_s} \left( \bR^3 \right)$ be defined by \eqref{eq:def-fst}. Let $O$ be a self-adjoint operator on $L^2 \left( \bR^3 \right) $ such that $\vertiii{O} < \infty$ as defined in \eqref{def:norm-O}. There exists a constant $C>0$ such that for $\kappa$ defined as
\begin{align}
\label{eq:kappas}
\kappa = C  \vertiii{O}  e^{e^{C t}}  
\end{align}
we have for all $0\leq \lambda\kappa \leq 1$
\begin{align}
 &\left\langle \Omega, \cW_N (t;0) e^{\lambda \sqrt{N} \phi_+  (q_t O \ph_t)/2} e^{2 \lambda \|{O} \| \cN_+ (t)} e^{\lambda \sqrt{N} \phi_+ (q_t O \ph_t)/2} \cW_N (t;0) \Omega \right\rangle  \notag\\
  &\hspace{1cm} \leq 
e^{ \kappa \left( N \lambda^3 \vertiii{O}^2+\lambda\right) }  \left\langle \Omega, e^{\lambda \sqrt{N} \phi_+ (f_{0;t})/2}  e^{\lambda \kappa \cN_+ (0)} e^{\lambda \sqrt{N} \phi_+ (f_{0;t})/2}  \Omega \right\rangle \; \label{eq:lemma-step2-upper}
\end{align}
and 
\begin{align}
 &\left\langle \Omega, \cW_N (t;0) e^{\lambda \sqrt{N} \phi_+  (q_t O \ph_t)/2} e^{-2 \lambda \| O \| \cN_+ (t)} e^{\lambda \sqrt{N} \phi_+ (q_t O \ph_t)/2} \cW_N (t;0) \Omega \right\rangle  \notag\\
  &\hspace{1cm} \geq 
e^{ -\kappa \left( N \lambda^3\vertiii{O}^2  + \lambda\right) }  \left\langle \Omega, e^{\lambda \sqrt{N} \phi_+ (f_{0;t})/2}  e^{- \lambda \kappa \cN_+ (0)} e^{\lambda \sqrt{N} \phi_+ (f_{0;t})/2}  \Omega \right\rangle \; . \label{eq:lemma-step2-lower}
\end{align} 
\end{lemma}
 
 \begin{proof}
The lower bound \eqref{eq:lemma-step2-lower} follows with ideas from \cite[Lemma 3.2]{KRS} and from Lemmas~\ref{lemma:Kj} and~\ref{lemma:f}. For $0\leq s \leq t$ and some (differentiable) $\kappa_s \geq 0$ with $\kappa_t = 2 \|O\|$, 
define the vector 
\begin{align}
\label{eq:xi-2}
\xi_t (s) = e^{-\lambda \kappa_{s} \cN_+ (s) /2} e^{\lambda \sqrt{N} \phi_+ (f_{s;t}) /2} \cW_N (s;0) \Omega \in \cF_{\perp \ph_s}^{\leq N} \,.
\end{align}
It satisfies 
\begin{align}
 \| \xi_t (0) \|^2 = \langle \Omega,  e^{\lambda \sqrt{N} \phi_+ (f_{0;t}) /2} e^{ -  \lambda \kappa_0 \cN_+ (0)} e^{\lambda \sqrt{N} \phi_+ (f_{0;t}) /2} \Omega \rangle\label{eq:xi-20}
 \end{align}
and 
\begin{align} \| \xi_t (t) \|^2 = \left\langle \Omega,  \cW_N (t;0)^* e^{\lambda \sqrt{N} \phi_+ (q_t O \ph_t) /2} e^{-2 \lambda \|O\| \cN_+ (t)} e^{\lambda \sqrt{N} \phi_+ (q_t O \ph_t) /2} \cW_N (t;0) \Omega \right\rangle \, .  \label{eq:xi-2t}
\end{align}
Note that the definition \eqref{eq:xi-2} is similar to the vector defined at the beginning of the proof of \cite[Lemma 3.2]{KRS}. The crucial difference is that, here, in \eqref{eq:xi-2}, for the lower bound,  the exponential of the number of particles operator comes with a negative constant in front (in contrast to a positive one in \cite{KRS} for the upper bound).

As in the proof of Lemma \ref{lemma:step1},  we want to control the difference of  \eqref{eq:xi-20} and \eqref{eq:xi-2t} through the derivative 
\begin{align}
\label{eq:J} 
\partial_s \| \xi_t (s) \|^2 = -2i  \, \Im \left\langle \xi_t (s) , \cJ_{N,t} (s) \xi_t (s) \right\rangle
 \end{align} 
where $\cJ_N(s)$ is (in the sense of a quadratic form on $\cF_{\perp \varphi_s}^{\leq N}$ as in \eqref{eq:form-LN}) given by 
\begin{align}
\label{eq:gen}
 \cJ_{N,t} (s) = \; & e^{-\lambda \kappa_{s} \cN_+  (s)/ 2} e^{\lambda \sqrt{N} \phi_+ (f_{s;t}) /2} \cL_N (s) e^{-\lambda \sqrt{N} \phi_+ (f_{s;t}) /2} e^{\lambda \kappa_{s} \cN_+ (s) / 2} \notag\\
 &+ e^{-\lambda \kappa_{s} \, \cN_+ (s) / 2} \left[ i \partial_s e^{\lambda \sqrt{N} \phi_+ (f_{s;t}) /2} \right] e^{-\lambda \sqrt{N} \phi_+ (f_{s;t}) /2} e^{\lambda \kappa_{s} \, \cN_+  (s) /2} -\frac{i\lambda}{2} {\dot\kappa_s} \, \cN_+ (s) \,,
 \end{align}
 where we denote $\dot \kappa_s = d\kappa_s/ds$. 
For this computation it is convenient to embed $\cF_{\perp \varphi_s}^{\leq N}$ into the full Fock space $\cF$ in which case $\cN_+(s)$ can be replaced by the $s$-independent $\cN$ (for more details see the discussion before \cite[Eq.~(3.3)]{KRS}).  We proceed as in \cite{KRS} and compute the anti-symmetric part of $\cJ_{N,t }(s)$ explicitly with the help of \cite[Propositions~2.2--2.4]{KRS}, and show that its norm is bounded by terms of order $N \lambda^3 $  and $\lambda$.  

To this end, recalling the definition of $\cL_N(s)$ in \eqref{eq:form-LN} and analogous calculations as in \cite[(3.4)--(3.5)]{KRS}, we have 
\begin{align}
&  e^{-\lambda \kappa_{s} \cN_+ (s)/ 2}  e^{\lambda \sqrt{N} \phi_+ (f_{s;t}) /2} d\Gamma (h_{\rm H} (s) + K_{1,s}) e^{-\lambda \sqrt{N} \phi_+ (f_{s;t}) /2} e^{\lambda \kappa_{s} \cN_+ (s) / 2}\notag \\ 
 & = \frac{i \lambda \sqrt{N}}{2}  \phi_- ((h_{\rm H} (s) + K_{1,s}) f_{s;t}) + T_1 + S_1  \label{eq:LN-1}
\end{align}  
where we introduced the notation $\phi_-(f) = -i (b(f)-b^*(f))$, $S_1 = S_1^*$ is symmetric and 
\begin{align}
T_1 & = i \sqrt{N} \left( \frac{\sigma_{\| h_{s;t} \|}}{\| h_{s;t} \|} \cosh (\lambda \kappa_{s} /2)  -1 \right) \phi_- ((h_{\rm H} (s) + K_{1,s}) h_{s;t})\notag  \\
 & \quad + i \sqrt{N} \frac{\sigma_{\| h_{s;t} \|}}{\| h_{s;t} \|} \frac{\gamma_{\| h_{s;t} \|} - 1}{\| h_{s;t} \|^2} \langle  h_{s;t} , (h_{\rm H} (s) + K_{1,s}) h_{s;t} \rangle \cosh (\lambda \kappa_{s} /2)  \phi_- (h_{s;t}) 
\end{align}
with $h_{s;t} = \lambda f_{s;t}/2$, as well as $\gamma_s = \cosh s$ and $\sigma_s = \sinh s$ as in the proof of Lemma~\ref{lemma:step1}. We use the bounds \eqref{eq:bounds-b} and 
\begin{align}
\label{eq:bound-h1}
\| h_{\rm H} (s) h_{s;t} \|_2 \leq C \| h_{s;t} \|_{H^2 \left( \bR^3 \right)} \leq C  \lambda \vertiii{ O } e^{ e^{C t}}  , \quad 
\| K_{1,s} h_{s;t} \|_2 \leq C \lambda \| O \|e^{ C t} 
\end{align}
for all $0\leq s\leq t$ by Lemmas~\ref{lemma:Kj} and~\ref{lemma:f}, and conclude that for all $0 \leq \lambda\kappa_s \leq 1$ and $s \in [0,t]$, we have 
$\| T_1 \| \leq  C e^{e^{C t}}  N \vertiii{O}^3  \lambda^3$.  

We proceed similarly with the remaining terms of \eqref{eq:gen}.  For the second term of $\cL_{N,t}(s)$ in \eqref{eq:form-LN}, we find with analogous calculations as the ones leading to \cite[Eq.~(3.6)]{KRS}, 
using that 
\begin{align}
\| K_{2,s} \| \leq \| K_{2,s} \|_{L^2( \bR^3 \times \bR^3)} \leq C, \quad \| f_{s;t} \|_2 \leq e^{C t } \| O \|
\end{align}
for  $0\leq s\leq t$ by Lemmas~\ref{lemma:Kj} and~\ref{lemma:f}, 
\begin{align}
\label{eq:LN-2} 
&   e^{-\lambda \kappa_{s} \cN_+ (s) / 2}  e^{\sqrt{N} \phi_+ (h_{s;t}) } \left( \frac{1}{2} \int \left[ \overline{K_{2,s} (x,y)} b_x b_y  + K_{2,s} (x,y) b_x^* b_y^* \right]  dx dy \right)  e^{- \sqrt{N} \phi_+ (h_{s;t}) } e^{\lambda \kappa_{s} \cN_+  (s) / 2}\notag \\ 
& = - \frac{i \lambda \sqrt{N}}{2} \phi_- (K_{2,s}\overline{ f_{s;t}}) + S_2 + T_2 + i R_2 
\end{align} 
where $S_2=S_2^*$ is symmetric, $T_2$ is bounded as $T_1$ above by $\| T_2 \| \leq  C e^{e^{C t}}  N \vertiii{O}^3  \lambda^3$, and $R_2$ contains all the remaining terms of order $\lambda$, which are given by
\begin{align}
R_2 & = -i \frac {\lambda  \kappa_{s}}2  \int \left[ \overline{K_{2,s} (x,y)} b_x b_y  - K_{2,s} (x,y) b_x^* b_y^* \right]  dx dy \notag \\ & 
\quad + i \frac {\lambda  \sqrt{N}}2 \left[ \left( 1 - \frac{\cN_+(s)+1/2}{N}\right) b (K_{2,s} \overline{f_{s;t}})  -b^* (K_{2,s} \overline{f_{s;t}})\left( 1 - \frac{\cN_+(s)+1/2}{N}\right)\right] \notag \\ &
\quad -i \frac \lambda {4\sqrt{N}}  \int \left[ \overline{K_{2,s} (x,y)} b^*(f_{s;t}) a_x a_y  - K_{2,s} (x,y) a_y^* a_x^* b(f_{s;t}) \right]  dx dy \notag \\ & 
\quad -i \frac \lambda {4\sqrt{N}}  \int \left[ \overline{K_{2,s} (x,y)} a^*(f_{s;t}) a_x b_y  - K_{2,s} (x,y) b_y^* a_x^* a(f_{s;t}) \right]  dx dy \,.
\end{align}
The bounds in  Lemmas~\ref{lemma:Kj} and~\ref{lemma:f} imply that\footnote{The corresponding bound in \cite{KRS} is incorrectly claimed with $\cN_+(s)$ instead of $\cN_+(s)+1$ on the r.h.s., resulting in a missing error term of order $\lambda$ which is independent of $N$, however, and hence irrelevant for $N\to \infty$. The same applies to the corresponding bounds on $R_3$, $R_4$ and $R_5$ below.}
\begin{align}
 R_2 \geq - C \lambda  \left( \kappa_s + \|O\| e^{Ct} \right)   \left( \cN_+ (s) + 1 \right)
\end{align}
for $0\leq s \leq t$. 

For the third term of  \eqref{eq:form-LN}, we proceed as in \cite[Eq.~(3.7)]{KRS} and use 
\begin{align}
\| K_{1,s}\| \leq C, \quad \| v* \vert \varphi_s \vert^2 \|_\infty \leq C , \quad \| f_{s;t}\|_2 \leq  e^{ C t } \| O \|
\end{align}
 from Lemmas~\ref{lemma:Kj} and~\ref{lemma:f} to conclude that 
\begin{align}
\notag
& e^{-\lambda \kappa_{s} \cN_+ (s) / 2}  e^{ \sqrt{N} \phi_+ (h_{s;t}) } 
 d\Gamma (v *|\ph_{s}|^2 + K_{1,s} - \mu_{s}) \frac{ \cN_+ (s) - 1}{2N}  
e^{- \sqrt{N} \phi_+ (h_{s;t}) } e^{\lambda \kappa_{s} \cN_+ (s)  / 2} \\ & =  S_3 + T_3 + i R_3\label{eq:LN-3} 
 \end{align} 
where $S_3 = S_3^*$ is symmetric, $\| T_3 \| \leq  C e^{e^{C t}}  N \vertiii{O}^3  \lambda^3$  and 
\begin{align}
  R_3 \geq - C e^{C  t } \| O \| \lambda (\cN_+ (s)+1) \,.
\end{align} 

For the forth term on the r.h.s. of \eqref{eq:form-LN}, we have with ${d}_s = \left( v * \vert \varphi_s \vert^2 \right) \varphi_s $
\begin{align}
& \frac 1{\sqrt{N}} e^{-\lambda \kappa_{s} \cN_+ (s) / 2}  e^{ \sqrt{N} \phi_+ (h_{s;t})} \left( \cN_+ (s)\, b ( {d}_s )+ b^* ( {d}_s )\, \cN_+ (s)  \right) e^{- \sqrt{N} \phi_+ (h_{s;t}) } e^{\lambda \kappa_{s} \cN_+ (s) / 2}\notag \\ 
 & =  S_4 + T_4 + i R_4 \label{eq:LN-4}
\end{align}
where $S_4$ and $R_4$ are symmetric and $\| T_4 \| \leq C N ( e^{C t} \| O \| + \kappa_{s})^3 \lambda^3$ for all $0\leq s  \leq t$ and $\lambda\kappa_s \leq 1$. The term $R_4$ equals   
\begin{align}
R_4  &  = \frac{ \lambda \kappa_{s}}{2 i \sqrt{N}} \left( \cN_+(s)  b({d}_s)  - b^*({d}_s) \cN_+(s) \right)  + i  \cN_+(s) ( \Im \langle {d}_s, h_{s;t} \rangle \left( 1 - \cN_+(s)/N \right)  \notag\\
 &\quad - i  \frac{\cN_+(s) }{N} \left( a^*(h_{s;t}) a ( {d}_s) - a^* ( {d}_s) a(h_{s;t})\right) +   \phi_-(h_{s;t} ) b({d}_s) +b^* ( {d}_s) \phi_-(h_{s;t} )  \,.
\end{align}
Since 
\begin{align}
\| {d}_s \|_2 \leq  \| v* \vert \varphi_s \vert^2 \|_\infty \| \varphi_s \|_2 \leq C , \quad \| f_{s;t} \|_2 \leq  e^{C t} \| O \|
\end{align}
by Lemmas~\ref{lemma:Kj} and~\ref{lemma:f}, we have 
\begin{align}
 R_4 \geq - C \lambda  \left( \kappa_s + \|O\| e^{Ct} \right)   \left( \cN_+ (s) + 1 \right)
\end{align}
 for all $0\leq s \leq t$.

Next,  we consider the fifth term on the r.h.s.  of \eqref{eq:form-LN} and follow the same strategy as the one leading to \cite[Eq.~(3.8)]{KRS}. With 
\begin{align}
\| v* \left( f_{s;t} \varphi_s \right) \|_\infty \leq C \| f_{s;t}\|_{2}  \| \varphi_s \|_{H^1(\bR^3)} \leq C e^{C t} \| O \|  \label{eq:bounds-4}
\end{align}
for  $0\leq s \leq t$ from \eqref{eq:estimate-phi-H1} and \eqref{eq:estimate-f-H2} we find that 
\begin{align}
&e^{-\lambda \kappa_{s} \cN_+ (s) / 2}  e^{\sqrt{N} \phi_+ (h_{s;t})} \int  v(x-y) \left[\varphi_s (x)  a^*_y a_x b_y + \overline{\varphi_s(x)} b_y^* a_x^* a_y \right]  dxdy e^{-\sqrt{N} \phi_+ (h_{s;t})} e^{\lambda \kappa_{s} \cN_+ (s) / 2} \notag\\ & = S_5 + T_5 + i R_5  \label{eq:LN-5}
\end{align} 
where $S_5^* = S_5$ is symmetric, $\| {T}_5 \|  \leq  C e^{Ct} N \| O \|^3 \lambda^3$ and 
\begin{align}
i R_5 &= - \frac{\lambda \kappa_{s}}{2\sqrt{N}} \int v(x-y) \, \left[\ph_s (y) b_x^* a_y^* a_x - \overline{\ph_s (y)} a_x^* a_y b_x \right] \; dx dy  \notag\\
&\quad -  \int  v(x-y)  \left[ h_{s;t} (y) \ph_s (y) b_x^* b_x - \overline{h_{s;t} (y)} \overline{\ph_s (y)} b_x^* b_x  \right]\; dx dy  \notag\\
&\quad-  \int v(x-y) \left[ h_{s;t} (x) \ph_s (y) b_x^* b_y^* - \overline{h_{s;t} (x)} \overline{\ph_s (y)}  b_y b_x  \right] \; dx dy \,  \notag\\
 &\quad+ \int v(x-y) \left[ \overline{h_{s;t} (x)} \ph_s (y) (1-\cN_+ (s)/ N) a_y^* a_x - h_{s;t} (x) \overline{\ph_s(y)} a_x ^* a_y (1-\cN_+(s) /N) \right] \; dx dy \notag \\
&\quad-  \frac{1}{N}  \int v(x-y) \left[ \ph_s (y) a_x^* a(h_{s;t} ) a_y^* a_x - \overline{\ph_s (y)} a_x^* a_y a^* (h_{s;t}) a_x \right] \; dx dy  \,.
\end{align} 
Using again \eqref{eq:bounds-b} and  \eqref{eq:bounds-4} as well as 
\begin{align}
\| v^2 * \vert \varphi_s \vert^2 \|_\infty \leq C \| \varphi_s \|_{H^1(\bR^3)}^2 \leq C
\end{align}
by \eqref{eq:ass-v}, 
we find  that 
\begin{align} 
 R_5 \geq - C \lambda  \left( \kappa_s + \|O\| e^{Ct} \right)   \left( \cN_+ (s) + 1 \right)
\end{align}
for all $0 \leq s \leq t $.  

Finally, we consider the last term on the r.h.s. of \eqref{eq:form-LN}, proceeding as in \cite[Eq.~(3.9)]{KRS}. With \eqref{eq:bounds-b} and
\begin{align}
\label{eq:f-h2}
\| v^2 * f^2_{s;t}\|_\infty \leq C \|f_{s;t}\|^2_{H^1(\bR^3)} \leq C e^{e^{Ct}} \vertiii{O} 
\end{align}
we find that 
\begin{align}
&\frac{1}{2N}  e^{-\lambda \kappa_{s} \cN_+ (s) / 2} e^{ \sqrt{N} \phi_+ (h_{s;t}) } \int dxdy \,  v(x-y) a^*_x a^*_y a_y a_x   e^{-  \sqrt{N} \phi_+ (h_{s;t})} e^{\lambda \kappa_{s} \cN_+ (s) / 2} \notag \\
 &= S_6 + T_6 + i R_6 \label{eq:LN-6}
 \end{align}  
with
\begin{align}
 i R_6  =\frac{\lambda}{2\sqrt{N}} \int v(x-y) \, \left[ \overline{f_{s;t} (y)}  a_x^* a_x b_y -    f_{s;t} (y) b^*_y a_x^* a_x \right] \; dxdy  \,.
 \end{align}  
Again  $S_6 = S_6^*$ is symmetric and $\| T_6 \| \leq  Ce^{e^{C t }}  N \vertiii{O}^3  \lambda^3$. Furthermore,  with \eqref{eq:f-h2} a Cauchy--Schwarz inequality yields 
\begin{align}
   R_6 \geq - Ce^{e^{Ct} }  \vertiii{O}  \lambda \cN_+ (s)  \,.
\end{align}

If we combine \eqref{eq:LN-1}, \eqref{eq:LN-2}, \eqref{eq:LN-3}, \eqref{eq:LN-4}, \eqref{eq:LN-5} and \eqref{eq:LN-6}, we conclude that the first term on the r.h.s. of \eqref{eq:gen} is given by 
\begin{align}
& e^{-\lambda \kappa_{s} \cN_+ (s) / 2} e^{\sqrt{N} \phi_+ (h_{s;t})} \cL_N (s) e^{-\sqrt{N} \phi_+ (h_{s;t})} e^{ \lambda \kappa_{s} \cN_+ (s) / 2} \notag\\
& = \frac{i \lambda \sqrt{N}}{2}  \phi_- ((h_{\rm H} (s) + K_{1,s} +  K_{2,s} J ) f_{s;t}) + S +T + i R  \label{eq:gen1}
\end{align} 
where $S^* = S$ is symmetric,  $\| T \| \leq    C e^{e^{C t} } N \vertiii{O}^3 \lambda^3$ and
\begin{align}
  R \geq -C \lambda (\vertiii{O} e^{e^{C t}  } +  \kappa_{s}) \left(  \cN_+ (s) +1 \right)
\end{align}
for all $0\leq s\leq t$ and  $0\leq \lambda \kappa_s \leq 1$.  
For the second term of the r.h.s. of \eqref{eq:gen} we find as in \cite[p.~2613]{KRS} using the definition of $f_{s;t}$ in \eqref{eq:def-fst} that 
\begin{align}
 & e^{-\lambda \kappa_{t-s} \cN_+ (s) / 2} \left[ i \partial_s e^{ \sqrt{N} \phi_+ (h_{s;t}) } \right] e^{- \sqrt{N} \phi_+ (h_{s;t}) } e^{\lambda \kappa_{t-s}  \cN_+ (s) /2} \notag\\
 &=  - \frac{i \lambda \sqrt{N}}{2} \,  \phi_- ( i \partial_s f_{s;t}) + \wt{S} + \wt{T}\notag \\
  &=  - \frac{i \lambda \sqrt{N}}{2} \,  \phi_- \left( \left( h_{\rm H} (s) + \widetilde{K}_{1,s} -  \widetilde{K}_{2,s} J  \right)  f_{s;t}\right)  + \wt{S} + \wt{T} \notag\\
    &=  - \frac{i \lambda \sqrt{N}}{2} \,  \phi_- \left( \left( h_{\rm H} (s) + K_{1,s} - K_{2,s} J  \right)  f_{s;t}\right)  + \wt{S} + \wt{T} \label{eq:lastline}
\end{align} 
where $\wt{S} = \wt{S}^*$ is symmetric and $\| \wt{T} \| \leq C e^{e^{Ct}} N \vertiii{O}^3 \lambda^3$.  We remark that the last equality holds as an identity in the sense of a quadratic form on $\cF_{\perp \varphi_s}^{\leq N}$ where the projection $q_s$ acts as the identity. 

With \eqref{eq:gen1} and  \eqref{eq:lastline}, we conclude that 
\begin{align}
  \frac{1}{i} \left[ \cJ_{N,t} (s) - \cJ_{N,t}^* (s) \right]   & \geq   -C e^{e^{C t} } N  \vertiii{O}^3 \lambda^3 - \lambda \left[  C (\vertiii{O} e^{e^{Ct}} +  \kappa_s) + \dot{\kappa}_s \right] \cN_+ (s)   \notag \\  &  \quad - C  \lambda (\vertiii{O} e^{e^{Ct}} +  \kappa_s)  \label{2rhs}
\end{align} 
for all $0\leq s \leq t $ and $0\leq \lambda \kappa_s \leq 1$.  
We shall choose 
$$
\kappa_s = 2 \|O\| e^{C (t-s) } +  \vertiii{O}  e^{e^{C t}} \left( e^{C(t-s)} - 1 \right)  
$$
in which case the second term on the r.h.s.\ of \eqref{2rhs} vanishes. With this choice of $\kappa$, we thus have
 from \eqref{eq:J}
\begin{align}
\partial_s \| \xi_t (s) \|^2  \geq - C e^{e^{C t}}  \left[ N \lambda^3 \vertiii{O}^3  + \lambda \vertiii{O}\right]  \| \xi_t (s) \|^2 
\end{align}  
for suitable $C>0$. With Gronwall's inequality, we arrive at
\begin{align}
 \| \xi_t (t) \|^2 \geq e^{-C ( N \lambda^3 \vertiii{O}^3 + \lambda \vertiii{O}) e^{e^{Ct}} t } \, \| \xi_t (0) \|^2  \,.
\end{align}
This concludes the proof of the lower bound.

As already mentioned at the beginning of the proof, the upper bound follows along the same lines. One simply replaces $\kappa_s$ by $-\kappa_s$ and estimates the various error terms $R_j$ for $2\leq j\leq 6$ from above instead of from below.
\end{proof}

The third step, Lemma \ref{lemma:step3}, is proven similarly to \cite[Lemma 3.3]{KRS}. 

 \begin{lemma}
 \label{lemma:step3}
There exists a constant $C_1 > 0$ such that for all  $t>0$, $0 \leq \kappa \leq C_1 \vertiii{O} e^{e^{C_1} t}$  and  $0\leq \lambda \leq   e^{-e^{C_1  t}} / (C_1 \vertiii{O} )$ 
\begin{align}
\label{eq:lemma-step3-upperbound} 
\ln \left\langle \Omega, e^{\lambda \sqrt{N} \phi_+ (f_{0;t})/2}  e^{\lambda \kappa \cN_+ (0)} e^{\lambda \sqrt{N} \phi_+ (f_{0;t})/2}  \Omega \right\rangle \leq  \frac{\lambda^2 N} 2 \| f_{0;t} \|^2 + C_1 N \lambda^3 \vertiii{O}^3 e^{e^{C_1 t}}   
\end{align}
and  
\begin{align}
\label{eq:lemma-step3-lowerbound} 
\ln \left\langle \Omega, e^{\lambda \sqrt{N} \phi_+ (f_{0;t})/2}  e^{-\lambda \kappa \cN_+ (0)} e^{\lambda \sqrt{N} \phi_+ (f_{0;t})/2}  \Omega \right\rangle \geq  \frac{\lambda^2 N}2  \| f_{0;t} \|^2 -C_1 N \lambda^3 \vertiii{O}^3 e^{e^{C_1 t}} \,.
\end{align}
\end{lemma}

\begin{proof} We start with the lower bound \eqref{eq:lemma-step3-lowerbound}, we proceed similarly as in the proof of the previous Lemmas. Following \cite[Lemma 3.3]{KRS}, we define for $s \in [0,1]$ the vector
\begin{align}\label{lef}
\xi_s =  e^{(1-s)^2 (N - 2 \cN_+(0)) \| h_t \|^2 /2} e^{-\lambda \kappa \cN_+ (0) /2} e^{s \sqrt{N} \phi_+ (h_t)} e^{(1-s) \sqrt{N} b^* (h_t)} e^{(1-s) \sqrt{N} b(h_t)}\Omega
\end{align}
where we introduced the notation $h_t = \lambda f_{0;t}/2 \in L^2_{\perp \ph} (\bR^3)$. Note that the last exponential factor in \eqref{lef} could be omitted since $b(h_t) \Omega = 0$, but it is actually convenient to keep it for the calculation of the derivative of $\partial_s \| \xi_s\|^2$. Compared to the upper bound in \cite{KRS}, we need the additional term $e^{-(1-s)^2\cN_+(0)\|h_t\|^2}$  in \eqref{lef}, as will be seen below.
We have 
\begin{align}
\label{eq:xi-31}
 \| \xi_1 \|^2 =  \left\langle \Omega, e^{\lambda \sqrt{N} \phi_+ (f_{0;t})/2}  e^{-\lambda \kappa \cN_+ (0)} e^{\lambda \sqrt{N} \phi_+ (f_{0;t})/2}  \Omega \right\rangle
\end{align}
and 
\begin{align}
\label{eq:xi-32} 
\| \xi_0 \|^2 = e^{N \| h_t \|^2} \langle e^{\sqrt{N} b^* (h_t)} \Omega , e^{- (\lambda \kappa + 2 \|h_t\|^2)  \cN_+ (0)} e^{\sqrt{N} b^* (h_t)} \Omega \rangle .
\end{align}
The latter quantity will lead to the desired bound on the r.h.s.\ of \eqref{eq:lemma-step3-lowerbound}.  In order to compare \eqref{eq:xi-31} and \eqref{eq:xi-32}, we compute the derivative of $\xi_s$  as
\begin{align}
 \partial_s \| \xi_s \|^2 = 2\, \Re\, \langle \xi_s , \cG_s \xi_s \rangle \end{align}
 where, following \cite[Eq.~(3.12) et seq.]{KRS}, 
 \begin{align}
 \cG_s & = 2 (1-s) \cN_+(0) \|h_t\|^2 \notag \\ 
 &\quad   - e^{-\lambda \kappa \cN_+ (0) /2} e^{s \sqrt{N} \phi_+ (h_t)} \left[ (1-s) \| h_t \|^2 \cN_+ (0) + (1-s) a^* (h_t) a(h_t) - \sqrt{N} \| h_t \|^2 (1-s)^2 b^* (h_t) \right] \notag\\
  &\hspace{4cm} \times e^{-s \sqrt{N} \phi_+ (h_t)} e^{\lambda \kappa_t \cN_+ (0) /2} \; .
 \end{align}
 Using that $ \| h_t \|_2 \leq \lambda \|O\| e^{C t }/2$
by Lemma \ref{lemma:f}, it follows from the calculation \cite[Eq.~(3.12)  et seq.]{KRS} that 
\begin{equation}
\cG_s  =  (1-s) \cN_+(0) \|h_t\|^2  - (1-s) a^*(h_s) a(h_s) + T
\end{equation}
with $\|T\| \leq C N \lambda^3 \|O\|^3 e^{C t}$ as long as $\lambda \kappa \leq 1$. Since $\cN_+(0)  \| h_t\|^2 \geq a^*(h_t) a(h_t)$, the remaining terms are positive, hence
\begin{align}
 \partial_s \| \xi_s \|^2  \geq - CN \lambda^3 \|O\|^3 e^{C t} \|\xi_s\|^2 \,.
  \end{align}
With Gronwall's inequality we arrive at 
\begin{align}
 \| \xi_1 \|^2 \geq e^{-CN\lambda^3 \|O\|^3 e^{Ct}}   \| \xi_0 \|^2\,.
\end{align} 

It remains to compute \eqref{eq:xi-32}.  To this end, let us introduce $\kappa' = \kappa + 2 \|h_t\|^2 /\lambda = \kappa + \lambda\|f_{0;t}\|^2/2$. As in \cite[Lemma 3.3]{KRS} we compute
\begin{align}
& e^{N \lambda^2 \| f_{0;t} \|^2/4}  \langle e^{\sqrt{N} \lambda b^* (f_{0;t})/2} \Omega , e^{-\lambda \kappa' \cN_+ (0) }e^{\sqrt{N} \lambda b^* (f_{0;t})/2} \Omega \rangle\notag\\
 &= e^{N \lambda^2 \| f_{0;t} \|^2/4} \sum_{n=0}^N \frac{N^n \lambda^{2n}}{4^{n} (n!)^2} e^{-\lambda \kappa' n} \| b^* (f_{0;t})^n \Omega \|^2 \; 
\end{align} 
and furthermore
\begin{align}
 & \| b^* (f_{0;t})^n \Omega \|^2 \notag
 \\ &= \left\| a^* (f_{0;t}) (1- \cN_+ (0) /N)^{1/2} a^* (f_{0;t}) (1-\cN_+ (0)/N)^{1/2} \cdots a^* (f_{0;t}) (1-\cN_+ (0)  /N)^{1/2} \Omega \right\|^2 \notag \\ 
 &= \frac{(N - (n-1)) \cdots (N-1)}{N^{(n-1)}} \| a^* (f_{0;t})^n \Omega \|^2 =  \frac{(N-1)! n! }{N^{(n-1)} (N-n)!}  \| f_{0;t} \|_2^{2n} \; .
\end{align}
Thus we have 
\begin{align}
e^{N \lambda^2 \| f_{0;t} \|^2/4}  & \langle e^{\sqrt{N}\lambda b^* (f_{0;t})/2} \Omega , e^{-\lambda \kappa' \cN_+ (0) }e^{\sqrt{N} \lambda b^* (f_{0;t})/2} \Omega \rangle \notag\\
&= e^{N \lambda^2 \| f_{0;t} \|_2^2/4} \sum_{n=0}^N {N \choose n} \frac{\lambda^{2n} \| f_{0;t} \|_2^{n} }{4^{n}} e^{-\lambda \kappa' n}  \notag\\
  &= e^{N \lambda^2 \| f_{0;t} \|_2^2/4} \, \left(1 + \frac{\lambda^{2} \| f_{0;t} \|_2^2}{4} e^{-\lambda \kappa'} \right)^N \notag\\
  &= e^{N \left( \lambda^2 \| f_{0;t} \|_2^2/4 + \ln \left( 1 +  \lambda^2 \| f_{0;t} \|_2^2 e^{-\lambda \kappa'} /4 \right) \right) }  \notag\\
  &\geq  e^{N \left( \lambda^2 \| f_{0;t} \|_2^2 ( 1 + e^{-\lambda \kappa'}) \right)/4 - N \lambda^4 \| f_{0;t} \|_2^4  / 32} 
\end{align}
where we used that $\ln(1+x) \geq x - x^2/2$ for $x\geq 0$. Using in addition that $e^{-\lambda\kappa'} \geq 1- \lambda \kappa'$ and $\|f_{0;t}\|_2 \leq \|O\| e^{Ct}$, we arrive at the desired bound \eqref{eq:lemma-step3-lowerbound}. 

The upper bound \eqref{eq:lemma-step3-upperbound} follows in essentially the same way, see \cite[Lemma 3.3]{KRS}. 
\end{proof}

\begin{proof}[Proof of Theorem \ref{thm:main}]
The upper bound \eqref{eq:thm-upper-bound} is an immediate consequence of \eqref{eq:start}, the upper bound in \eqref{eq:lemma-step1-upper}, \eqref{eq:lemma-step2-upper} and \eqref{eq:lemma-step3-upperbound}. Similarly, the lower bound \eqref{eq:thm-lower-bound} follows by combining \eqref{eq:start} with the lower bound in \eqref{eq:lemma-step1-upper}, \eqref{eq:lemma-step2-lower} and \eqref{eq:lemma-step3-lowerbound}. 
 \end{proof}

\subsection*{Acknowledgments} The authors thank G\'erard Ben Arous for pointing out the question of a lower bound.  Funding from the European Union's Horizon 2020 research and innovation programme under the ERC grant agreement No. 694227 (R.S.) and under the Marie Sk\l{}odowska-Curie  Grant Agreement No. 754411 (S.R.) is gratefully acknowledged.


\end{document}